\documentclass[11pt,a4paper]{article}
\usepackage[margin=1in]{geometry}
\usepackage[utf8]{inputenc}
\usepackage{soul}
\usepackage{url}
\usepackage[colorlinks=true]{hyperref}
\usepackage[small]{caption}
\usepackage{graphicx}
\usepackage{amsmath,amssymb,amsthm}
\usepackage{mathtools}
\usepackage{multirow}
\usepackage{xspace}
\usepackage{bm}
\usepackage{enumerate}
\usepackage{enumitem}
\usepackage[boxed,linesnumbered]{algorithm2e}	

\allowdisplaybreaks
\newtheorem{proposition}{Proposition}

\newtheorem{definition}{Definition}[section]

\newtheorem{lemma}[definition]{Lemma}

\newtheorem{theorem}{Theorem}

\newcommand{\bigO}{\mathcal{O}}
\newcommand{\Prob}[2]{\mathbf{P}_{#1} \left( #2 \right)}
\newcommand{\Expec}[2]{\mathbf{E}_{#1} \left[ #2 \right]}

\newcommand{\tcons}{\ensuremath{T^{\mathrm{cons}}}}
\newcommand{\thit}{\ensuremath{T^{\mathrm{hit}}}}

\newcommand{\model}{\ensuremath{\mathcal{M}}}

\newcommand{\nmut}{\ensuremath{z}} 
\newcommand{\bw}{\mathbf{w}} 
\newcommand{\bx}{\mathbf{x}} 
\newcommand{\bz}{\mathbf{z}}

\newcommand{\bp}{\mathbf{p}} 
\newcommand{\bone}{\mathbf{1}}

\renewcommand{\leq}{\leqslant}
\renewcommand{\le}{\leqslant}

\renewcommand{\ge}{\geqslant}

\title{\textbf{On a Voter Model with Context-Dependent Opinion Adoption}}

\author{
Luca Becchetti\\
		{\small{}Sapienza University of Rome}\\
		{\small{}Rome, Italy}\\
		{\small{}\texttt{becchetti@diag.uniroma1.it}}\\
\and Vincenzo Bonifaci\\
		{\small{}Roma Tre University}\\
		{\small{}Rome, Italy}\\
		{\small{}\texttt{vincenzo.bonifaci@uniroma3.it}}\\
\and Emilio Cruciani \\
		{\small{}Paris-Lodron University of Salzburg}\\
		{\small{}Salzburg, Austria}\\
		{\small{}\texttt{emilio.cruciani@plus.ac.at}}\\
\and Francesco Pasquale\\ 
		{\small{}University of Rome ``Tor Vergata''}\\
		{\small{}Rome, Italy}\\
		{\small{}\texttt{pasquale@mat.uniroma2.it}} 
}

\begin{document}

\maketitle 

\begin{abstract}
Opinion diffusion is a crucial phenomenon in social networks, often underlying
the way in which a collective of agents develops a consensus on relevant
decisions.  The voter model is a well-known theoretical model to study opinion
spreading in social networks and structured populations. Its simplest version
assumes that an updating agent will adopt the opinion of a neighboring agent
chosen at random. The model allows us to study, for example, the probability
that a certain opinion will fixate into a consensus opinion, as well as the
expected time it takes for a consensus opinion to emerge. 

Standard voter models are oblivious to the opinions held by the agents involved
in the opinion adoption process. We propose and study a context-dependent
opinion spreading process on an arbitrary social graph, in which the
probability that an agent abandons opinion $a$ in favor of opinion $b$ depends
on both $a$ and $b$. We discuss the relations of the model with existing voter
models and then derive theoretical results for both the fixation probability
and the expected consensus time for two opinions, for both the synchronous and
the asynchronous update models.

\end{abstract}

\section{Introduction}
The \emph{voter model} is a well-studied stochastic process defined on a graph
to model the spread of opinions (or genetic mutations, beliefs, practices,
etc.) in a population~\cite{Liggett:1985,hassin01voter}. In a voter model, each
node maintains a state, and when a node requires updating it will import its
state from a randomly chosen neighbor. Updates can be asynchronous, with one
node activating per step~\cite{Liggett:1985}, or synchronous, with all nodes
activating in parallel~\cite{hassin01voter}.  While the voter model on a graph
has been introduced in the 1970s to model opinion dynamics, the case of a
complete graph is also very well-known in population genetics where, in fact,
it was introduced even earlier, to study the spread of mutations in a
population~\cite{ewens2012mathematical, nowak2006evolutionary}. 

Mathematically, among the main quantities of interest in the study of voter
models, there are the \emph{fixation probability} of an opinion---the
probability of reaching a configuration in which each node adopts such
opinion---and the expected \emph{consensus} (or \emph{absorption})
\emph{time}---the expected number of steps before all nodes agree on an
opinion. Such quantities could in principle be computed for any $n$-node graph
by defining a Markov chain on a set of $C^n$ configurations, where $C$ is the
number of opinions, but such an approach is computationally infeasible even for
moderate values of $n$.  Therefore, a theoretical analysis of a voter process
will often focus on obtaining upper and lower bounds for these quantities,
still drawing heavily on the theory of Markov
chains~\cite{aldous2002reversible, Levin:2009}, but with somewhat different
approaches and tools for the synchronous and asynchronous cases. 

A limitation of the standard voter process is that the dynamics is oblivious to
the states of both the agent $u$ that is updating and of the neighbor that $u$
copies its state from, and the copying always occurs. One could easily imagine
a situation (for example, in politics) where an agent holding opinion $a$ is
more willing to adopt the opinion $b$ of a neighbor rather than to adopt
opinion $c$; in general, the probability of abandoning opinion $a$ in favor of
opinion $b$ might depend on both $a$ and $b$. This motivates the study of
\emph{biased} voter models~\cite{Berenbrink:2016, sood2008voter} and in
particular motivates us to introduce a voter model with an opinion adoption
probability that depends on the context, that is, on the opinions of
\emph{both} agents involved in an opinion spreading step. 

We define and study extensions of the voter model that allow the opinion
adoption probability to depend on the pair of opinions involved in an update
step. We consider both an asynchronous variant and a synchronous variant of a
context-dependent voter model with two opinions, 0 and 1. We assume that an
agent holding opinion $c \in \{0,1\}$ is willing to copy the opinion of an
agent holding opinion $c' \in \{0,1\}$ with some probability $\alpha_{c,c'}$,
which models the bias in the update.  We study both the fixation probabilities
and the expected consensus time. 

\subsection{Our Findings}
In general, a seemingly minor feature as the form of bias we consider has a
profound impact on the analytical tractability of the resulting model.  While
the unbiased case\footnote{The model is unbiased when
$\alpha_{0,1}=\alpha_{1,0}$.} can still be connected to a variant of the voter
model and analyzed accordingly with some extra work, the same is not possible
for the biased case. Specifically, in Section~\ref{se:unbiased} we prove that a
lazy variant of the voter model is equivalent (i.e., it produces the same
distribution over possible system's configurations) to the unbiased variant of
the model we consider. The proof, given in Lemma~\ref{prop:equiv-models} for
the synchronous case\footnote{The proof for the asynchronous case is
essentially the same up to technical details and is omitted for the sake of
space.}, uses a coupling between the Markov chains that describe the two
models. For the asynchronous case, this allows us to directly leverage known
connections between the asynchronous voter model and random walks
(Proposition~\ref{prop:async-unbiased}). For the case of the clique, this
general result can be improved, providing explicit, tight bounds on expected
consensus time (Theorem~\ref{thm:exact-clique}).  In the synchronous case, the
above connection is not immediate (i.e., Proposition~\ref{prop:async-unbiased}
does not apply) and analyzing expected time to consensus requires adapting
arguments that have been used for continuos Markov chains to the synchronous,
discrete setting (Theorem~\ref{cl:consensus_1}).

The biased case is considerably harder to analyze, the main reason being that
it is no longer possible to collapse the Markov chain describing the system
(whose state space is in general the exponentially large set of all possible
configurations) to a ``simpler'' chain, e.g., a random walk on the underlying
network, not even in the case of the clique.  Despite these challenges, some
trends emerge from specific cases.  Interestingly, it is possible to derive the
exact fixation probabilities for the class of regular networks, highlighting a
non-linear dependence from the bias (Theorem~\ref{thm:async-biased-clique-fp}),
while an asymptotically tight analysis for the clique
(Theorem~\ref{thm:async-biased-clique-ect}) suggests that the presence of a
bias may have a positive impact on achieving faster consensus in dense
networks. Though seemingly intuitive, this last aspect is not a shared property
of biased opinion models in general~\cite{montanari2010spread,
anagnostopoulos22biased}. The behavior of the model is considerably more
complex and technically challenging in the synchronous, biased case. In
particular, the preliminary results we obtain highlight a general dependence of
fixation probabilities (Proposition~\ref{prop:fixation-clique-sync}) and,
notably, expected consensus time (Theorem~\ref{th:biased-sync-time}) from both
the bias and the initial configuration. 

\subsection{Related work}
For the sake of space, we mostly discuss results that are most closely related
to the setting we consider.

\paragraph{Voter and voter-like models.} Due to its versatility, the voter
model has been defined multiple times across different disciplines and has a
vast literature. As mentioned in the introduction, the special case of the
voter model on a complete graph was first introduced in mathematical genetics,
being closely related to the so-called Wright-Fisher and Moran
processes~\cite{Moran:1958, kimura1962prob, ewens2012mathematical}. The first
asynchronous formulation of a voter model on a connected graph has been
proposed in the probability and statistics community in the
1970s~\cite{Liggett:1985, donnelly1983finite}, while Hassin and
Peleg~\cite{hassin01voter} were the first to study this model in a synchronous
setting. A classic result of these papers is that the fixation probability of
an opinion $c$ is equal to the weighted fraction of nodes holding opinion $c$,
where the weight of a node is given by its degree~\cite{hassin01voter,
sood2008voter}.  The expected consensus time of the voter model is much more
challenging to derive exactly, even for highly structured graphs. In the
asynchronous case, it has often been studied by approximating the process with
a continuous diffusion partial differential
equation~\cite{ewens2012mathematical, sood2008voter, baxter2008fixation}. For
two opinions on the complete graph, this yields the approximation 
\begin{equation}
\label{eq:basicvoter-async-clique}
T(n) \approx n^2 h(k/n)
\end{equation}
where $T(n)$ is the expected consensus time on the $n$-clique, $k$ is the
number of nodes initially holding the first opinion and $h(p) = -p \ln p -(1-p)
\ln (1-p)$~\cite{sood2008voter, ewens2012mathematical}.  To the best of our
knowledge, however, no error bound was known for such diffusion
approximations\footnote{For the $n$-clique, we show
that~\eqref{eq:basicvoter-async-clique} is correct within an additive
$\bigO(n)$ term. See Theorem~\ref{thm:exact-clique}
(Section~\ref{sec:async-consensus}).}. Another approach is to use the duality
between voter model and coalescing random walks~\cite{Liggett:1985,
donnelly1983finite}, which involves no approximations, but the resulting
formulas are hard to interpret, and to paraphrase Donnelly and
Welsh~\cite{donnelly1983finite}, ``an exact evaluation of the expected
absorption time for a general regular graph is a horrendous computation''.  As
for approximations, expected consensus time for the voter model can be bounded
by $\bigO((d_{\textrm{avg}}/d_{\textrm{min}})(n/\Phi))$, where
$d_{\textrm{avg}}$ and $d_{\textrm{min}}$ are, respectively, the average and
minimum degrees~\cite{Berenbrink:2016}.  Most relevant to our discussion is the
biased voter model considered by Berenbrink et al.~\cite{Berenbrink:2016}, in
which the probability of adoption of an alternative opinion $c'$ depends on
$c'$ (and only on $c'$). While our model is different if we consider more than
$2$ opinions, there are several other differences with respect
to~\cite{Berenbrink:2016} even in the binary case. In particular, Berenbrink et
al.~only consider the synchronous setting, they assume there is a ``preferred''
opinion that is never rejected and that there is a constant gap between the
adoption probabilities of the preferred opinion and of the non-preferred one.
Finally, they only consider the case where the number $k$ of nodes initially
holding the preferred opinion is at least $\Omega(\log n)$.  Thus, for example,
their results do not apply in the neutral case ($\alpha_{01}=\alpha_{10}$) or
when $k$ is, say, $\sqrt{\log n}$. Our results for the biased, synchronous
case (Theorem~\ref{th:biased-sync-time}) are complementary to those
in~\cite{Berenbrink:2016}.  While their results are stronger when the above
assumptions hold, ours address the general and challenging case of an arbitrary
initial configuration. 

\paragraph{Pull vs push.} We only reviewed here models where nodes ``pull'' the
opinion from their neighbor, since both the standard voter model and our
generalization follow this rule, but we remark that ``push'' models, also known
as invasion processes, have also been defined and studied on connected
graphs~\cite{Lieberman:2005, diaz2016absorption}. The asynchronous push model
is sometimes called the (generalized) Moran process~\cite{diaz2016absorption,
nowak2006evolutionary}. We remark that while the pull and push models are
interchangeable on regular graphs, on irregular graphs their behavior can be
markedly different. 

\paragraph{Other biased opinion dynamics.} We are aware of only a few
analytically rigorous studies of biased opinion dynamics, including biased
variants of the voter model~\cite{sood2008voter, Berenbrink:2016,
anagnostopoulos22biased, CrucianiMQR21, durocher2022invasion}, sometimes framed
within an evolutionary game setting~\cite{montanari2010spread}.  In general,
these contributions address different models, be it because of the way in which
bias in incorporated within the voting rule, the opinion dynamics itself or the
temporal evolution of the process (e.g., synchronous vs asynchronous). We
remark that all these aspects can deeply affect the overall behavior of the
resulting dynamics. Specifically, as observed in a number of more or less
recent contributions~\cite{anagnostopoulos22biased, CrucianiMQR21,
cooper2018discordant, hindersin2014counterintuitive}, even minor changes in the
model that would intuitively produce consistent results with a given baseline
can actually induce fundamental differences in the overall behavior, so that it
is in general hard to predict if and when results for one model more or less
straightforwardly carry over to another model, even qualitatively.  Less
related to the spirit of this work, a large body of research addresses biased
opinion dynamics using different approaches, based on approximations and/or
numerical simulations.  Examples include numerical simulations for large and
more complex scenarios~\cite{hindersin2014counterintuitive}, mean-field or
higher-order~\cite{peralta2021effect} and/or continuous
approximations~\cite{assaf2012metastability}.  While these approaches can
afford investigation of richer and more complex evolutionary game settings
(e.g.,~\cite{peralta2021effect}), they typically require strong simplifying
assumptions to ensure tractability, so that it is harder (if not impossible) to
derive rigorous results.

\section{Model formulation}\label{se:vincenzo_model}

\paragraph{Notation.} For a natural number $k$, let $[k] :=
\{0,1,2,\ldots,k-1\}$. If $G = (V, E)$ is a graph, we write $N_G(u)$ (or simply
$N(u)$ if $G$ can be inferred from the context) for the set of neighbors of
node $u$ in $G$. We write $d_u$ for the degree of node $u$. 

\paragraph{Model.} We define an opinion dynamics model on networks. The
parameters of the model are: i) an underlying \emph{topology}, given by a graph
$G$ on $n$ nodes, with symmetric adjacency matrix $A=(a_{uv})_{u, v \in [n]}$;
ii) a number of \emph{opinions} (or \emph{colors}) $C \ge 2$; iii) an
\emph{opinion acceptance matrix} $(\alpha_{c,c'})_{c,c' \in [C]}$. The initial
opinion of each agent (node) $u$ is encoded by some $x_u^{(0)} \in [C]$. 

\SetAlgoNoEnd
\begin{algorithm}[t]
\caption{\FuncSty{Update}($u$)\label{algo:update}}
\DontPrintSemicolon
			\KwSty{Sample} $v \in N(u)$ \;
			$c \leftarrow x_u, c' \leftarrow x_v$\;
			\KwSty{Sample} $\theta \in [0,1]$ \;
			\If{$\theta < \alpha_{c,c'}$}{
				$x_u \leftarrow x_v$\;
				\Return{\FuncSty{accept}}\;
			}
			\Return{\FuncSty{reject}}\;
\end{algorithm}

For any node $u \in [n]$, we define an update process $\FuncSty{Update}(u)$
consisting of the following steps (summarized in Algorithm~\ref{algo:update}):

\begin{enumerate}[noitemsep,topsep=0pt]

\item \textbf{Sample}: Sample a neighbor $v$ of $u$ uniformly at random, i.e.,
according to the distribution $(a_{u1}/d_u,\ldots,a_{un}/d_u)$ where $a_{uv}=1$
if $u$ and $v$ are adjacent, $a_{uv}=0$ otherwise. Here $d_u=|N(u)|=\sum_{v \in
[n]} a_{uv}$ is the degree of node $u$. 

\item \textbf{Compare}: Compare $u$'s opinion $c = x_u$ with $v$'s opinion $c'
= x_v$. 

\item \textbf{Accept/reject}: With probability $\alpha_{c,c'}$, set $x_u
\leftarrow x_v$;  in this case we say $u$ \emph{accepts} $v$'s opinion.
Otherwise, we say $u$ \emph{rejects} $v$'s opinion.

\end{enumerate}

We consider two variants of the model, differing in how the updates are
scheduled.  In one iteration of the \emph{asynchronous} variant, $u \in [n]$ is
sampled at random and $\FuncSty{Update}(u)$ is applied.  In one iteration of
the \emph{synchronous} variant, each node $u \in [n]$ applies
$\FuncSty{Update}(u)$ in parallel.  We denote by $x_u^{(t)}$ the random
variable encoding the opinion of node $u$ after $t$ iterations of either the
synchronous or the asynchronous dynamics (depending on the context). 

The acceptance probabilities $\alpha_{c,c'}$ are parameters of the model. We
note that the parameters $\alpha_{c,c'}$ with $c=c'$ are irrelevant for the
dynamics, since a node sampling a neighbor of identical opinion will not change
opinion, irrespective of whether it accepts the neighbor's opinion or not.
Hence, to specify the opinion acceptance matrix $C(C-1)$ parameters are
sufficient; we can assume that the diagonal entries equal, say, 1. In
particular, when $C=2$ it is enough to specify $\alpha_{01}$ and $\alpha_{10}$.
When $\alpha_{01}=\alpha_{10}=1$, the model boils down to the standard voter
model~\cite{hassin01voter, Liggett:1985}. 

In the rest of this work we assume $C=2$. In this case, we say that the model
is \emph{unbiased} if the opinion acceptance matrix is symmetric, i.e.,
$\alpha_{01} = \alpha_{10}$, and \emph{biased} otherwise. 

\paragraph{Quantities of interest.} The \emph{fixation probability} of opinion
1 is the probability that there exists an iteration $t$ such that $x_u^{(t)}=1$
for all $u \in [n]$.  The \emph{consensus time} is the index of the first
iteration $t$ such that $x_u^{(t)}=x_v^{(t)}$ for all $u, v \in [n]$.

\section{The unbiased setting}\label{se:unbiased}
Before embarking on the biased case, which is substantially more complex, in
this section we review or prove directly results for the unbiased setting
($\alpha_{01}=\alpha_{10}$). We consider the asynchronous and the synchronous
variants separately. In all formulas of this section,
$\alpha=\alpha_{01}=\alpha_{10}$. 

\subsection{Asynchronous variant}\label{subse:unbiased_asynch}
The main, intuitive observation about the unbiased asynchronous variant of our
model is that the model can equivalently be described by a suitable, ``lazy''
voter model, where each iteration is either an idle iteration (with probability
$1-\alpha$) or an iteration of the standard asynchronous voter model (with
probability $\alpha$). 

This in turn implies that, for the fixation probability one can simply
disregard the idle iterations and therefore obtain the same fixation
probability as for the standard asynchronous voter model. In an arbitrary
topology, this was derived by Sood et al.~\cite{sood2008voter}: if we call
$\phi^{\mathrm{avoter}}$ the fixation probability of the asynchronous voter
model, then
\begin{equation}\label{eq:sood}
\phi^{\mathrm{avoter}} = \frac{\sum_{u \in [n]} d_u x^{(0)}_u}{\sum_{u\in [n]} d_u}, 
\end{equation}
where $x_u^{(0)}$ and $d_u$ are respectively the initial opinion and the degree
of node $u$.  Since $x_u^{(0)} \in \{0,1\}$, the fixation probability
$\phi^{\mathrm{avoter}}$ is proportional to the volume of nodes initially
holding opinion $1$
\footnote{We note incidentally that the fixation
probability can also be computed by suitably relating the asynchronous model to
the transition matrix of the lazy random walk we discuss in
Section~\ref{subse:unbiased_synch}. This connection is only mentioned here and
made rigorous in Appendix~\ref{se:async_lazy}}.

In the analysis of the expected consensus time, instead, one cannot ignore the
idle iterations, but since they occur with probability $1-\alpha$ independently
of other random choices, their effect is simply that of slowing down the
standard asynchronous voter process by a factor $1/\alpha$. This intuitive
argument can be formalized through a standard Markov chain coupling argument
(see Appendix~\ref{app:lazy-chains}). 

\begin{proposition}\label{prop:async-unbiased}
In the unbiased asynchronous case, the fixation probability is the same as for
the standard asynchronous voter model. The expected consensus time is
$T^{\mathrm{avoter}}/\alpha$, where $T^{\mathrm{avoter}}$ is the expected
consensus time of the standard asynchronous voter model. 
\end{proposition}

\paragraph{On the $n$-clique topology.}\label{sec:async-consensus}
As discussed in the introduction, a very natural question is: how large is the
expected consensus time on the $n$-clique as a function of $n$? Despite this
question having been studied multiple times before, in the literature there are
either diffusion approximations with unknown
error~\cite{ewens2012mathematical}, or exact formulas involving multiple
partial summations that are hard to interpret asymptotically~\cite{Glaz:1979}.
By further analyzing a result of Glaz~\cite{Glaz:1979}, we derive here an
explicit formula with bounded error that is easy to interpret, which in fact
agrees with the diffusion approximation up to lower order terms, thus also
showing that at least in this case, the diffusion
approximation~\eqref{eq:basicvoter-async-clique} yields a correct estimate. 

In the unbiased asynchronous model, the expected consensus time on the clique
is the same as the mean absorption time of the underlying birth-death process,
the state of which is summarized by the number of nodes holding opinion 1.
Call $T_k(n)$ the expected consensus time when starting from a configuration
with $k$ nodes holding opinion $1$, and the remaining $n-k$ holding opinion
$0$. We prove that $T_k(n) = \bigO(n^2/\alpha)$, more precisely:

\begin{theorem}
\label{thm:exact-clique}
If $\alpha_{01}=\alpha_{10}=\alpha$ (for some $\alpha>0$), 
then for each $k=1,\ldots,n-1$, 
\[
T_k(n) = \frac{1}{\alpha} n^2 h(k/n) + \bigO(n/\alpha),
\] 
where $h(p) := -p \ln p - (1-p) \ln (1-p)\le\ln 2$. 
\end{theorem}
\begin{proof}
On an $n$-clique, the process is equivalent to a birth-and-death
chain~\cite{Levin:2009} on $n+1$ states $0,1,2,\ldots,n$ (representing the
number of nodes with opinion, say, $1$).  Let us define the following
quantities: 
\begin{itemize}[noitemsep,topsep=0pt]
\item $p_k = \alpha_{01} k(n-k)/n(n-1)$ is the probability that the number of
nodes holding opinion 1 increases from $k$ to $k+1$ when $0 \le k < n$,
\item $q_k = \alpha_{10} k(n-k)/n(n-1)$ is the probability that the number of
nodes holding opinion 1 decreases from $k$ to $k-1$ when $0 < k \le n$. 
\end{itemize}
Note that $p_k=q_k$ for all $k$ due to the assumption $\alpha_{01}=\alpha_{10}$. 
Define the vector $T(n)$ as $T(n) = (T_1(n),\ldots,T_{n-1}(n))^\top$ and
consider the matrix
\[
B\!=\! \left(\begin{array}{ccccc}
p_1+q_1 & -p_1 & 0 & \ldots & 0 \\
-q_2 & p_2+q_2 & -p_2 & \ldots & 0 \\
\vdots & \ddots & \ddots & \ddots & \vdots \\
0 & \ldots & 0 & -q_{n-1} & p_{n-1} + q_{n-1} \\
\end{array}\right)\!. 
\]
The matrix $B$ is constructed so that $B T(n) = \mathbf{1}$, 
where $\mathbf{1}$ is the all-1 vector. This holds because of the recurrence
\[
	T_k(n) = 1 + (1 - p_k - q_k) T_k(n) + q_k T_{k-1}(n) + p_k T_{k+1}(n)
\]
for the mean consensus times.
Therefore, $T(n) = B^{-1} \mathbf{1}$. The matrix $B$ can be explicitly
inverted thanks to its tridiagonal structure; an explicit computation (see
Appendix~\ref{app:inversion}) yields
\begin{small}
\[
T_k(n) = \frac{n-1}{\alpha} \left( (n-k) (H_{n-1} - H_{n-k}) + k (H_{n-1} - H_{k-1}) \right), 
\]
\end{small}
where $H_k$ is the $k$-th harmonic number, $H_k = \sum_{j=1}^k 1/j$.  Recalling
the asymptotic expansion $H_n = \ln n + \gamma + \bigO(1/n)$, where $\gamma$ is
the Euler-Mascheroni constant, 
\begin{small}
\begin{align*}
    & T_k(n)
    = \frac{n-1}{\alpha} \!\left(
        (n-k) (H_{n} - H_{n-k})
        + k (H_{n} - H_{k})
    \right) \!+ \bigO\Big(\frac{n}{\alpha}\Big)
    \\
    &= \frac{n(n-1)}{\alpha} \!\left(
        \!\left( 1-\frac{k}{n} \right) \ln \frac{n}{n-k}
        + \frac{k}{n} \ln \frac{n}{k}
    \right) + \bigO(n/\alpha)
    \\
    &= \frac{n^2}{\alpha} h(k/n) + \bigO(n/\alpha),
\end{align*}
\end{small}
where $h(p) = -p \ln p - (1-p) \ln(1-p)$, which is such that 
$0 \le h(p) \leq \ln 2$ for every $p \in [0,1]$.
\end{proof}

\subsection{Synchronous variant}\label{subse:unbiased_synch}
The analysis of the synchronous variant in the unbiased setting relies 
on the tight connection between the unbiased case of the opinion 
dynamics we consider and (lazy) random walks on networks. 

\paragraph{Connections to lazy random walks.} We next provide an equivalent
formulation of our model, which reveals an interesting and useful connection to
lazy random walks. To this purpose, consider the following, alternative
dynamics, in which the behavior of the generic node $u$ at each iteration is
the following:
\begin{itemize}[noitemsep,topsep=0pt]
	\item Node $u$ independently tosses a coin with probability of ``heads'' 
	equal to $\alpha$;
	\item If ``heads'', $u$ samples a neighbor $v$ u.a.r.\ and copies
	$v$'s opinion; otherwise $u$ does nothing and keeps her opinion.
\end{itemize}

Let us call $\model_1$ the synchronous model described in
Section~\ref{se:vincenzo_model} and $\model_2$ the dynamics described above.
Then $\model_1$ and $\model_2$ are equivalent in the sense that, if they start
from the same initial state, they generate the same probability distribution
over all possible configurations of the system at any iteration $t$.
Intuitively speaking this is true since in $\model_1$ each node first samples a
neighbor and then it decides whether or not to copy its opinion according to
the outcome of a coin toss, while in $\model_2$ each node first tosses a coin
to decide whether or not to copy the opinion of one of the neighbors and then
it samples the neighbor. Since the outcome of the coin toss and the choice of
the neighbor are independent random variables, they produce the same
distribution on the new opinion of the node when commuted. In
Appendix~\ref{sec:coupling_unbiased_sync} we formalize the above equivalence
and we prove it by appropriately coupling the two processes using an inductive
argument.

Model $\model_2$ is interesting, since it describes a (lazy) voter model. As
such (and as we explicitly show in the proof of Theorem~\ref{cl:consensus_1}),
it is equivalent, in a probabilistic sense, to $n$ lazy, coalescing random
walks on the underlying network. This connection allows us to extrapolate the
probability of consensus to a particular opinion and to adapt techniques that
have been used to analyze the consensus time of the standard voter
model~\cite{hassin01voter, aldous2002reversible}.

\begin{theorem}\label{cl:consensus_1}
Assume model $\model_2$ starts in a configuration in which all nodes of a
subset $W \subset V$ have opinion~$1$ and all other nodes have opinion $0$. Let
$\phi$ and $\tcons$ denote \emph{fixation probability} (of opinion 1) and
\emph{time to consensus}, resp. Then:
\textsc{(i)}
	$\phi = ({\sum_{u\in W}d_u})/({\sum_{u \in V} d_u})$,
\textsc{(ii)}
	$ \Expec{}{\tcons}\le \beta_n \thit$,
where $\thit$ is the maximum expected hitting time associated with the graph
and $\beta_n = \bigO(1)$ when $\alpha\le 1/2$, while $\beta_n = \ln n + 3$ when
$\alpha > 1/2$.
\end{theorem}

\begin{proof}[Sketch of the proof]
We here only give a short idea of the proof and we defer a full-detailed proof
to Appendix~\ref{apx:unbiased_sync_thm_proof}.

The proof of \textsc{(i)} follows from the observation that, if we call
$\bp(t)$ the vector $\bp(t) = (p_1(t), \dots, p_n(t))$ where $p_i(t)$ is the
probability that node $i$ has opinion $1$ at round $t$ conditional on the
configuration at the previous round $\bx^{(t-1)}$, then for every round $t$ it
holds that $\Expec{}{\bx^{(t+1)} \,|\, \bx^{(t)}} = \bp(t+1) = P \bx^{(t)}$
where $P$ is the transition matrix of a lazy simple random walk on the
underlying graph. Iterating the above equality we have that $\lim_{t\rightarrow
\infty}\Expec{}{\bx^{(t)} \,|\, \bx^{(0)}} = \pi^\intercal \bx^{(0)} \bone$,
where $\pi$ is the stationary distribution of the random walk. Finally, the
formula for $\phi$ follows from the fact that, for each node $i$,
$\lim_{t\rightarrow \infty}\Expec{}{x_i^{(t)} \,|\, \bx^{(0)}}$ equals the
probability that node $i$ ends up with opinion $1$ and from the fact that the
stationary probability of a simple random walk being on a node $i$ is
proportional to the degree of $i$.

The proof of \textsc{(ii)} is an adaptation to our (discrete) case of the proof
strategy for the continuous case described
in~\cite[Section~14.3.2]{aldous2002reversible}: we leverage on the relation
between the convergence time of $\model_2$ and the maximum \textit{meeting
time} of two lazy random walks and, by using an appropriate martingale, we show
that the maximum meeting time is upper bounded by the maximum hitting time (see
Lemma~\ref{le:TvsM}).
\end{proof}

\section{The biased setting}\label{se:biased}
Without loss of generality, in the rest of this section we assume 
$\alpha_{01} \neq \alpha_{10}$ and we let $r = \alpha_{01}/\alpha_{10}$. 
In general, the biased setting is considerably harder to address, since 
the connection between our model and lazy random walks no longer 
applies in this setting, nor does it seem easy to track the evolution 
of the expected behavior of the model in a way that is mathematically 
useful. 

\subsection{Asynchronous variant}\label{subse:biased_asynch}
In the asynchronous case, we give a result for the fixation probability holding
for regular graphs, thanks to an equivalence with the fixation probability for
the $n$-clique (see~Appendix \ref{sec:async-fixation-reg}). We also bound the
expected consensus time in the specific case of the $n$-clique.  The
asynchronous variant of the process on the $n$-clique is equivalent to a
birth-and-death chain on $n+1$ states $0,1,2,\ldots,n$ representing the number
of nodes with opinion $1$. Now, however, the transition probabilities will not
be symmetric. In fact, the transition probabilities can be specified by $\{p_k,
q_k, r_k\}_{k=0}^n$ where $p_k+q_k+r_k = 1$, and:

\begin{itemize}[noitemsep,topsep=0pt]
\item $p_k$ is the probability that the number of nodes holding opinion 1
increases from $k$ to $k+1$ when $0 \le k < n$,
\item $q_k$ is the probability that the number of nodes holding opinion 1
decreases from $k$ to $k-1$ when $0 < k \le n$,
\item $r_k$ is the probability that the number of nodes holding opinion 1
remains $k$ when $0 \le k \le n$.
\end{itemize}
Due to our definition of the opinion dynamics, we have $p_0 = q_n = 0$ and, for
$0 < k < n$, 
\begin{align*}
p_k &= \left(1 - \frac{k}{n} \right) \cdot \frac{k}{n-1} \cdot \alpha_{01} =
\alpha_{01} \frac{k(n-k)}{n(n-1)}, \\
q_k &= \frac{k}{n} \cdot \frac{n-k}{n-1} \cdot \alpha_{10} = \alpha_{10}
\frac{k(n-k)}{n(n-1)}, \\
r_k &= 1 - (\alpha_{01}+\alpha_{10}) \frac{k(n-k)}{n(n-1)}.
\end{align*}

\begin{theorem}
\label{thm:async-biased-clique-fp}
Let $r = \alpha_{01}/\alpha_{10}$ and let $\phi_k$ be the fixation probability
of opinion $1$ on a regular $n$-nodes graph starting from a state in which $k$
nodes hold opinion $1$. Then, for $r \notin \{0,1\}$, 
\begin{equation}
\label{eq:fixprob-async}
\phi_k = \frac{1 - r^{-k}}{1 - r^{-n}}. 
\end{equation}
\end{theorem}
\begin{proof}
Thanks to the equivalence of the fixation probability between a regular
$n$-nodes graph and an $n$-clique (see Appendix~\ref{sec:async-fixation-reg}),
and by the analysis of a general birth-death process (see for
example~\cite[Section 6.2]{nowak2006evolutionary}), we get
\[
\phi_k = \frac{1 + \sum_{i=1}^{k-1} \prod_{j=1}^i \gamma_j}{1 + \sum_{i=1}^{n-1} \prod_{j=1}^i \gamma_j},
\]
where $\gamma_j = q_j/p_j = \alpha_{10} / \alpha_{01} = 1/r$ for all $j$. Hence
\[
\phi_k = \frac{1 + \sum_{i=1}^{k-1} r^{-i}}{1 + \sum_{i=1}^{n-1} r^{-i}} = 
\frac{1 + \frac{1 - r^{-k}}{1 - r^{-1}} - 1}{1 + \frac{1 - r^{-n}}{1 - r^{-1}} - 1} =
\frac{1 - r^{-k}}{1 - r^{-n}}. 
\]
Note that when $r \to 1$, we can evaluate $\phi_k$ by applying L'H\^{o}pital's
rule to~\eqref{eq:fixprob-async} and get $\phi_k=k/n$, which is consistent with
the results for the unbiased setting. When $r = 0$, nodes can only switch from
opinion $1$ to opinion $0$, so clearly $\phi_k=0$; this is also consistent with
\eqref{eq:fixprob-async} when $r \to 0$. 
\end{proof}

It is interesting to note that the expression from~\eqref{eq:fixprob-async}
coincides with the fixation probability in the standard Moran process
\cite[Chapter 6]{nowak2006evolutionary} when mutants (say, nodes with opinion
1) have a \emph{relative fitness} equal to $\alpha_{01}/\alpha_{10}$, and in
the initial configuration there are $k$ mutants out of $n$ nodes. In other
words, on an $n$-nodes regular graph the ratio $\alpha_{01}/\alpha_{10}$ can be
interpreted as a \emph{fitness} of sorts, even though there is no notion of
fitness or selection built in our model (recall that nodes are activated
uniformly at random). 

For the $n$-clique we are also able to bound the expected consensus time.
While the structure of the proof is similar to the one of
Theorem~\ref{thm:exact-clique}, the proof itself is considerably more involved
in the asymmetric setting, leading to qualitatively different results---namely,
an $O(n \log n)$ instead of an $O(n^2)$ worst case bound (see
Appendix~\ref{sec:biased-async-appendix}). 

\begin{theorem}\label{thm:async-biased-clique-ect}
If $T_k(n)$ is the expected consensus time in the $n$-clique when starting from
a state with $k$ nodes holding opinion 1, and $\alpha_{01},\alpha_{10}$ are
constants, then for each $k=1,\ldots,n-1$, 
\[
	T_k(n) = O(n \log n), 
\]
and for some values of $k$ the above bound is tight. 
\end{theorem}

\subsection{Synchronous variant}\label{subse:biased_synch}
In order to bound the fixation probabilities, denote by $\bx^{(t)} \in\{0,
1\}^n$ the state of the system at time $t$.  Conditioned on the state vector
$\mathbf x^{(t)}$, the probability that $x^{(t+1)}_u=1$ can be expressed as
follows: 

\begin{align*}
&\Prob{}{x^{(t+1)}_u = 1 \,|\, \bx^{(t)}}\\ 
&=\begin{cases} 
1 - \alpha_{10} \left(1 - \frac{ \sum_{v \in V} a_{uv} x_v^{(t)} }{d_u} \right) & \text{if } x_u^{(t)} = 1 \\
\alpha_{01} \frac{ \sum_{v \in V} a_{uv} x_v^{(t)} }{d_u} & \text{if } x_u^{(t)} = 0.
\end{cases}
\end{align*}
This follows since the probability that node $u$ samples a neighbor 
with opinion 1 is $\sum_v a_{uv} x_v^{(t)}/d_u$ and: 
\begin{itemize}[noitemsep,topsep=0pt]
\item when $x_u^{(t)}=1$, then $x_u^{(t+1)}=1$ iff either $u$ samples a
neighbor with opinion $1$, or $u$ samples a neighbor with opinion 0 and does not
accept its opinion (these two events are disjoint); 
\item when $x_u^{(t)}=0$, then $x_u^{(t+1)}=1$ iff $u$ samples a neighbor with
opinion $1$ and accepts its opinion. 
\end{itemize}
 
We did not exploit the graph topology so far. In the case of the $n$-clique (with 
loops, to simplify some expressions), let $k^{(t)}$ be the number of nodes 
with opinion 1 at time $t$. Specializing the formulas derived above we get
\[
\Prob{}{x^{(t+1)}_u=1 \,|\, \bx^{(t)}} = 
\begin{cases} 
1 - \alpha_{10} \left(1 - \frac{ k^{(t)} }{n} \right) & \text{if } x_u^{(t)} = 1 \\
\alpha_{01} \frac{ k^{(t)} }{n} & \text{if } x_u^{(t)} = 0.
\end{cases}
\]
Note that the expression above depends only on $k^{(t)}$ and $x_u^{(t)}$, 
and not on the entire state $\bx^{(t)}$. 
The process is thus equivalent to sampling, at each step $t$, 
$k^{(t)}$ Bernoulli random variables (r.v.)~with parameter $\beta_k := 
1 - \alpha_{10}(1-k^{(t)}/n)$, and $n-k^{(t)}$ bernoulli r.v.~with 
parameter $\gamma_k := \alpha_{01} k^{(t)}/n$. Collectively, the 
outcomes of these r.v.~constitute the new state $\bx(t+1)$.
Then, 
\begin{small}
	\begin{align*}
	\Expec{}{k^{(t+1)} \,|\, k^{(t)}} 
	&= (n - k^{(t)}) \alpha_{01} \frac{k^{(t)}}{n} + k^{(t)} \left(1 - \alpha_{10}( 1- \frac{k^{(t)}}{n} ) \right)
	\end{align*}
\end{small}
which, posing $y^{(t)}=k^{(t)}/n$, can be written as 
\begin{equation}
\label{eq:condex-y}
\Expec{}{y^{(t+1)} \,|\, y^{(t)}} = y^{(t)} + (\alpha_{01}-\alpha_{10}) y^{(t)} (1-y^{(t)}).
\end{equation}

\begin{proposition}\label{prop:fixation-clique-sync}
Assume $\alpha_{01}\le \alpha_{10}$. Then the fixation probability of opinion
0 is at least the fraction of agents holding opinion 0. 
\end{proposition}
\begin{proof}
Under the assumption $\alpha_{01}\le \alpha_{10}$, 
\begin{align*}
&\Expec{}{y^{(t+1)}} = \Expec{}{\Expec{}{y^{(t+1)} \,|\, y^{(t)}}} \\
&= \Expec{}{y^{(t)}} + (\alpha_{01}-\alpha_{10}) \Expec{}{y^{(t)}(1-y^{(t)})} 
\le \Expec{}{y^{(t)}}. 
\end{align*}
Hence, the succession $ (\Expec{}{y^{(t)}})_t$ is monotone and bounded and
attains a limit. This limit must coincide with the fixation probability,
because $y^{(t)}$ converges in distribution to a bernoulli random variable
$y^{(\infty)}$ and
$\Expec{}{y^{(\infty)}}=\Prob{}{y^{(\infty)}=1}=\Prob{}{\exists t: y^{(t)}=1}$
equals the fixation probability. Since $\Expec{}{y^{(0)}} = y^{(0)} =
k^{(0)}/n$, the fixation probability of opinion 1 must be at most $k^{(0)}/n$,
so that of opinion 0 is at least $1-k^{(0)}/n$. 
\end{proof}
Regarding the expected consensus time, we show the following by using the
technique of \emph{drift analysis}~\cite{lengler2018drift}. 
\begin{theorem}
\label{th:biased-sync-time}
If $T_k(n)$ is the expected consensus time in the $n$-clique when starting from a configuration with $k$ nodes holding opinion 1, and $\alpha_{01}=\alpha_{10}-\epsilon$, then
\[
T_k(n) \le \frac{n k}{\epsilon (n-1)}. 
\]
In particular, $T_k(n) \le \min( 2k/\epsilon, n/\epsilon)$ for each $k=1,2,\ldots,n-1$. 
\end{theorem}
\begin{proof}
We adapt a proof of~\cite[Theorem 2.1]{lengler2018drift} to our setting, since
their result is not suitable for systems with more than one absorbing state. In
the remainder, we assume $\alpha_{01} < \alpha_{10}$, we let $\epsilon =
\alpha_{10} - \alpha_{01}$, and we let $\nmut^{(t)} = k^{(t)}/n$, i.e.,
$\nmut^{(t)}$ is the fraction of agents with opinion $1$ at time $t$.  We begin
by defining the following stopping time:
\[
	T := \inf\{t\ge 0: \nmut^{(t)}\in\{0, 1\}\}.
\]
This definition is akin to the one given in~\cite[Theorem
2.1]{lengler2018drift}, but it accounts for the presence of two absorbing
states in the Markov chain defined by $\nmut^{(t)}$. Moreover, $\nmut^{(t)} =
z^{(t-1)}$ for every $t > T$, since $\nmut^{(t)}$ does not change after
absorption (regardless of the absorbing state).  We next note that for every
$t$, $\nmut^{(t)}\in\mathcal{S} = \left\{0,\frac{1}{n},\ldots , 1 -
\frac{1}{n}, 1\right\}$. Moreover, for every $s\in\mathcal{S}$ we have
$\Expec{}{\nmut^{(t+1)}\mid \nmut^{(t)} = s} = s - \epsilon s(1 - s)$, whence:
\[
	\Expec{}{\nmut^{(t)} - \nmut^{(t+1)}\mid \nmut^{(t)} = s} = \epsilon 
	s(1 - s),
\]
with the last quantity at least $\epsilon\frac{1}{n}\left(1 - 
\frac{1}{n}\right)$ for $s\in\mathcal{S} \setminus \{0, 1\}$. Next:
\begin{align*}
	&\Expec{}{\nmut^{(t+1)}\mid T > t}\\ 
	&= \sum_{s=1}^{n-1}\Expec{}{\nmut^{(t+1)}\mid \nmut^{(t)} = \frac{s}{n}}\cdot\Prob{}{\nmut^{(t)} 
	= \frac{s}{n}\mid T > t},
\end{align*}
where the first equality follows since, for $s\not\in\{0, 1\}$, 
$\nmut^{(t)} = s$ implies $T > t$. Similarly to the proof of \cite[Theorem 
2.1]{lengler2018drift}, the equality above implies
\begin{equation}\label{eq:diff}
	\Expec{}{\nmut^{(t)} - \nmut^{(t+1)}\mid T > 
	t}\ge\epsilon\frac{1}{n}\left(1 - \frac{1}{n}\right).
\end{equation}
We let $\delta = \epsilon\frac{1}{n}\left(1 - \frac{1}{n}\right)$ for 
conciseness. We next have:
\begin{align*}
	\Expec{}{\nmut^{(t)}} &\stackrel{(a)}{=} \Expec{}{\nmut^{(t)}\mid T > t}\Prob{}{T > t} +
		\\&\qquad + \Prob{}{\nmut^{(t)} = 1\mid T\le t}\cdot\Prob{}{T\le t}
\end{align*}
and therefore
\begin{align*}
	&\Expec{}{\nmut^{(t+1)}} \stackrel{(a)}{=} \Expec{}{\nmut^{(t+1)}\mid T > t}\cdot\Prob{}{T > t} +
		\\&\qquad + \Prob{}{\nmut^{(t+1)} = 1\mid T\le t}\cdot\Prob{}{T\le t}\\
	&\stackrel{(b)}{\le} \left(\Expec{}{\nmut^{(t)}\mid T > t} - \delta\right)\cdot\Prob{}{T > t} +
		\\&\qquad + \Prob{}{\nmut^{(t+1)} = 1\mid T\le t}\cdot\Prob{}{T\le t}\\
	&\stackrel{(c)}{=} \Expec{}{\nmut^{(t)}} - \Prob{}{\nmut^{(t)} = 1\mid T\le t}\cdot\Prob{}{T\le t} +
		\\&\qquad + \Prob{}{\nmut^{(t+1)} = 1\mid T\le t}\cdot\Prob{}{T\le t} - \delta\cdot\Prob{}{T > t}.
\end{align*}

In the derivations above, $(a)$ simply follows from the law of total
probability, considering that $T\le t$ implies $\nmut^{(t)}\in\{0, 1\}$, $(b)$
follows from~\eqref{eq:diff}, while $(c)$ follows by replacing the equation of
$\Expec{}{\nmut^{(t)}}$ into the last step of the derivation. Next, we note
that
\[
	\Prob{}{\nmut^{(t+1)} = 1\mid T\le t} = \Prob{}{\nmut^{(t)} = 1\mid T\le t}
\]
by definition of the $\nmut^{(t)}$, whence we obtain:
\begin{equation}\label{eq:diff2}
	\delta\cdot\Prob{}{T > t}\le\Expec{}{\nmut^{(t)}} - 
	\Expec{}{\nmut^{(t+1)}}.
\end{equation}
Now, observe that \eqref{eq:diff2} is exactly \cite[(2.4) in Theorem 
2.1]{lengler2018drift}. From this point, the proof proceeds exactly as in 
\cite[(2.4) in Theorem 2.1]{lengler2018drift}, so that we finally have:
\[
	\Expec{}{T}\le\frac{z^{(0)}}{\delta} = \frac{n k}{\epsilon(n-1)},
\]
if at time $t = 0$ we have $k$ agents with opinion $1$.
\end{proof}

\section{Conclusions and Outlook}
Natural directions for future work include considering more opinions 
and general topologies.

\paragraph{More opinions.} The case of more opinions presents no major
challenges in the unbiased case, both in its asynchronous and synchronous
variants, something we did not discuss for the sake of space. In this case, one
can simply focus on one opinion at a time, collapsing the remaining opinions
into an ``other'' class.  Proceeding this way, it is easy to extend the results
we presented in Section~\ref{se:unbiased} to the general case: for $k > 2$
opinion, the fixation probability for opinion $i$ is $\frac{\sum_{u\in
W}d_u}{\sum_ud_u}$, where $W$ is the subset of nodes with opinion $i$ in the
initial configuration.  The biased case is considerably harder and the
technical barriers are twofold: one is the general difficulty of characterizing
the expected change of the global state in the biased setting even in the case
of 2 opinions (see next paragraph). The other is the possible presence of
rock-paper-scissors like dynamics that may arise depending on the distribution
of the opinion biases.

\paragraph{General topologies.} As also suggested by previous work, albeit for
different models~\cite{montanari2010spread, anagnostopoulos22biased, LGP22}, we
believe the biased case might give rise to diverse and possibly
counterintuitive behaviors. In general, a crucial technical challenge is
characterizing the evolution of the global state across consecutive steps,
since this in general depends on the current configuration in a way that is
highly topology-dependent and hard to analyize. Some recent
results~\cite{schoenebeck2018consensus, shimizu2020quasi} proposed techniques
relying on variants of the expander mixing lemma to investigate quasi-majority
dynamics on expanders. Unfortunately, these techniques do not obviously extend
to the biased voter models we consider. Indeed and interestingly, the class of
dynamics these techniques apply to does not even include the standard voter
model as a special case.

\medskip\noindent In general, we believe that extending and/or improving our
results for the biased setting might require refining important techniques,
such as those of~\cite{schoenebeck2018consensus, shimizu2020quasi} or the ones
discussed in~\cite{lengler2018drift}.

\section*{Acknowledgements}
\begin{itemize}
\item Partially supported by the ERC Advanced Grant 788893 AMDROMA
``Algorithmic and Mechanism Design Research in Online Markets'', the EC
H2020RIA project ``SoBigData++'' (871042), and the MIUR PRIN project ALGADIMAR
``Algorithms, Games, and Digital Markets.
\item Work of the second author was carried out in association with Istituto di
Analisi dei Sistemi ed Informatica, Consiglio Nazionale delle Ricerche, Italy.
\item Supported by the Austrian Science Fund (FWF): P 32863-N. 
\item This project has received funding from the European Research Council
(ERC) under the European Union's Horizon 2020 research and innovation programme
(grant agreement No 947702).
\item This project was partly supported also by Rome Technopole, PNRR grant M4-C2-Inv. 1.5 CUP B83C22002820006 to VB. In particular, conceptualization and manuscript editing were funded by Rome Technopole.

\end{itemize}

\bibliographystyle{plain} 
\bibliography{moran}

\newpage
\begin{center}
\begin{LARGE}
\textbf{Appendix}
\end{LARGE}
\end{center}

\appendix
\section{Unbiased asynchronous model: connections to lazy random 
walks}\label{se:async_lazy}
For every node $v \in V$, the expected state of $v$ at time $t+1$, conditioned
on $\bx^{(t)} = \bx$ is
\begin{align*}
	\Expec{}{x_v^{(t+1)} \;\big|\; \bx^{(t)}=\bx} 
	&=\left(1 - \frac{1}{n}\right )x_v + \frac{1}{n} \left((1 - 
	\alpha)x_v + \frac{\alpha}{d_v}\sum_{u \in N(v)} x_u\right) = 
	\left(1 - \frac{\alpha}{n}\right )x_v  + \frac{\alpha}{nd_v}\sum_{u 
	\in N(v)} x_u.
\end{align*}
In vector form:
\[
	\Expec{}{\bx^{(t)} \;\big|\; \bx^{(t-1)}=\bx} = \left(1 - 
	\frac{\alpha}{n}\right )\bx + \frac{\alpha}{n}P\bx,
\]
where $P = D^{-1}A$ is the transition matrix of the simple random walk on $G$
(with $D$ the diagonal degree matrix and $A$ the adjacency matrix of the graph)
and $I$ is the identity matrix. From this we obtain:
\[
	\Expec{}{\bx^{(t)}} = \left(1 - 
	\frac{\alpha}{n}\right )\bx^{(t-1)} + \frac{\alpha}{n}P\bx^{(t-1)} 
	= \hat{P}\bx^{(t-1)},
\] 
where $\hat{P} := \left(1 - \frac{\alpha}{n}\right)I + 
\frac{\alpha}{n}P$ is a row-stochastic matrix.\footnote{This trivially 
follows since the entries of each rows of $P$ sum to $1$.} We finally obtain
\[
	\Expec{}{\bx^{(t)}} = \hat{P}^t \bx^{(0)}.
\]
The matrix $\hat{P}$ corresponds to an ergodic Markov chain, whose left 
and right eigenvectors are the same of the transition matrix $P = D^{-1}A$ of the 
random walk on the underlying graph $G$. In particular, the main left 
eigenvector is the stationary distribution of the random walk on $G$, 
which corresponds to \eqref{eq:sood} in Section \ref{se:unbiased}.

\section{Proof of Proposition~\ref{prop:async-unbiased}}
\label{app:lazy-chains}

For completeness' sake in this section we give a full proof of
Proposition~\ref{prop:async-unbiased}. The proof makes use of some standard
relations between a Markov chain and its corresponding \textit{lazy} version.

\subsection*{Lazy Markov chains}
Let $P$ be a square row-stochastic matrix\footnote{A matrix whose entries are
in the interval $[0,1]$ and such that the entries of each row sum up to $1$},
let $\{X_t\}_t$ be a Markov chain with transition matrix $P$, let $\alpha \in
(0,1)$, and let $P_1 = (1-\alpha) I + \alpha P$, where $I$ is the identity
matrix with the same dimension of $P$. Observe that $P_1$ is itself a
stochastic matrix and that the Markov chain with transition matrix $P_1$
proceeds as follows: At each step, with probability $1-\alpha$ the chain stays
where it is, and with probability $\alpha$ it does one step according to
transition matrix $P$. We call the Markov chain $\{Y_t\}_t$ with transition
Matrix $P_1$ the \textit{lazy version} of Markov chain $\{X_t\}_t$ with
parameter $\alpha$.

Let $\{X_t\}_t$ and $\{Y_t\}_t$ be the Markov chains with transition matrices
$P$ and $P_1$, respectively. We can define the Markov chain $\{Y_t\}_t$ on the
same probability space of $\{X_t\}_t$ using an extra random source as follows:
Let $\{B_t\}_t$ be a sequence of independent and identically distributed
Bernoulli random variables such that $B_t = 0$ with probability $1-\alpha$ and
$B_t = 1$ with probability $\alpha$, let $\sigma_t = \sum_{i=1}^t B_t$, for
every $t = 1, 2, \dots$, and let $\{Y_t\}$ be defined as follows 
\begin{equation}\label{eq:ycoupled}
\left\{
\begin{array}{ccl}
Y_0 & = & X_0 \\
Y_{t} & = & X_{\sigma_t} \qquad \mbox{ for } t = 1, 2, \dots
\end{array}
\right.
\end{equation}
Observe that, from the coupling in~\eqref{eq:ycoupled} it follows that, if
chain $\{X_t\}$ visits the sequence of states $(x_0, x_1, x_2, \dots)$ then
also the chain $\{Y_t\}_t$ visits the same states, in the same order, remaining
on each state for $1/\alpha$ units of time in expectation. More formally, for a
state $x$ let $\tau^Y(x)$ be the first time the chain $\{Y_t\}_t$ hits state
$x$. From the construction of the coupling it follows that, for each sequence
of states $(x_0, x_1, \dots, x_t)$ such that $P(x_i, x_{i+1}) > 0$ for every $i
= 0, \dots t-1$, it holds that
\begin{equation}\label{eq:yconsensus}
\Expec{}{\tau^Y(x_t) \,|\, (X_0, X_1, \dots, X_t) = (x_0, x_1, \dots, x_t)} = t
/ \alpha
\end{equation}
Hence, if we have a theorem about the hitting time of some state (or set of
states) for the original chain $\{X_t\}_t$, it directly applies also to the
lazy chain, with a $1/\alpha$ multiplicative factor.

\subsection*{Proof of Proposition~\ref{prop:async-unbiased}}
Let $\Omega = \{0,1\}^n$ be the state space of all possible configurations for
a set of $n$ nodes where each node can be in state either $0$ or $1$. For an
arbitrary graph $G$, let $\{X_t\}_t$ be the Markov chain with state space
$\Omega$ describing the voter model on $G$ and let $\{Y_t\}_t$ be the Markov
chain with state space $\Omega$ describing our unbiased asynchronous variant.
Notice that, according to the definition of our unbiased asynchronous variant,
$\{Y_t\}_t$ is the lazy version of $\{X_t\}_t$ with parameter $\alpha$.  Using
the coupling in~\eqref{eq:ycoupled} it thus can be defined on the same
probability space of $\{X_t\}_t$ in a way that chain $\{X_t\}_t$ ends up with
all nodes in state $1$ if and only if chain $\{Y_t\}_t$ ends up with all nodes
in state $1$. Hence, the fixation probability in our unbiased asynchronous
model is equal to the fixation probability in the voter model. Moreover,
observe that the \textit{consensus} time is the hitting time of the set of two
states $\{(0, \dots, 0), (1, \dots, 1)\} \subseteq \{0,1\}^n$. Hence, according
to~\eqref{eq:yconsensus}, the expected consensus time of our unbiased
asynchronous model equals the expected consensus time of the voter model
multiplied by $1/\alpha$. \qed

\section{Computation of $T_k(n)$ in Theorem~\ref{thm:exact-clique}}
\label{app:inversion}

\begin{lemma}
\label{lem:B-inverse}
If $T(n) = B^{-1} \bone$, then for each $k=1,\ldots,n-1$, 
\[
T_k(n) = \frac{n-1}{\alpha} \left( (n-k) (H_{n-1} - H_{n-k}) + k (H_{n-1} - H_{k-1}) \right). 
\]
\end{lemma}
\begin{proof}
After inverting the matrix $B$ as in~\cite[Eq.~(5)]{Glaz:1979}, 
\begin{equation}
\label{eq:glaz-formula}
    T_k(n) 
    = \frac{
        \sum_{s=1}^{n-1} \left(
            \sum_{\ell=1}^{m(s,k)} \prod_{i=1}^{\ell-1} q_i \prod_{j=\ell}^{s-1} p_j
        \right) \left(
            \sum_{\ell=1}^{M(s,k)} \prod_{i=s+1}^{n-\ell} q_i \prod_{j=n-\ell+1}^{n-1} p_j
        \right)
    }{
        \sum_{\ell=1}^{n} \prod_{i=1}^{n-\ell} q_i \prod_{j=n-\ell+1}^{n-1} p_j
    },
\end{equation}
where $m(s,k) := \min(s,k)$, and $M(s,k) := n - \max(s,k)$. 
We already observed that $p_i = r q_i$ for all $i=1,\ldots,n-1$, where $r =
\alpha_{01}/\alpha_{10}$ (we do not yet substitute $r=1$ because the following
computation will also be useful in the biased setting). Hence,
\begin{align}
    T_k(n) 
    &= \frac{
        \sum_{s=1}^{n-1} \left(
            \sum_{\ell=1}^{m(s,k)} \prod_{i=1}^{\ell-1} q_i \prod_{j=\ell}^{s-1} r q_j
        \right) \left(
            \sum_{\ell=1}^{M(s,k)} \prod_{i=s+1}^{n-\ell} q_i \prod_{j=n-\ell+1}^{n-1} r q_j
        \right)
    }{
        \sum_{\ell=1}^{n} \prod_{i=1}^{n-\ell} q_i \prod_{j=n-\ell+1}^{n-1} r q_j
    }
    \nonumber
    \\
    &= \frac{
        \sum_{s=1}^{n-1} \left(
            \sum_{\ell=1}^{m(s,k)} r^{s-\ell} \prod_{i=1}^{s-1} q_i
        \right) \left(
            \sum_{\ell=1}^{M(s,k)} r^{\ell-1} \prod_{i=s+1}^{n-1} q_i
        \right)
    }{
        \sum_{\ell=1}^{n} r^{\ell-1} \prod_{i=1}^{n-1} q_i
    }
    \nonumber
    \\
    &= \frac{
        \sum_{s=1}^{n-1} \left(
            \prod_{i=1}^{s-1} q_i \cdot \prod_{i=s+1}^{n-1} q_i 
        \right) \left(
            \sum_{\ell=1}^{m(s,k)} r^{s-\ell}
        \right) \left(
            \sum_{\ell=1}^{M(s,k)} r^{\ell-1}
        \right)
    }{
        \prod_{i=1}^{n-1} q_i \cdot \sum_{\ell=1}^{n} r^{\ell-1}
    }
    \nonumber
    \\
    &= \frac{
        \prod_{i=1}^{n-1} q_i \cdot \sum_{s=1}^{n-1} q_s^{-1} \left(
            \sum_{\ell=1}^{m(s,k)} r^{s-\ell}
        \right) \left(
            \sum_{\ell=1}^{M(s,k)} r^{\ell-1}
        \right)
    }{
        \prod_{i=1}^{n-1} q_i \cdot \sum_{\ell=1}^{n} r^{\ell-1}
    }
    \nonumber
    \\
    &= \frac{
        \sum_{s=1}^{n-1} q_s^{-1} \left(
            \sum_{\ell=1}^{m(s,k)} r^{s-\ell}
        \right) \left(
            \sum_{\ell=1}^{M(s,k)} r^{\ell-1}
        \right)
    }{
        \sum_{\ell=1}^{n} r^{\ell-1}
    }. 
    \label{eq:formula-orribile}
\end{align}
Using now the assumption of the unbiased setting, $r=
\alpha_{01}/\alpha_{10}=1$ and the two inner summations in the numerator
simplify to $m(s,k)$ and $M(s,k)$ respectively.  After substituting $q_s =
\alpha \frac{s(n-s)}{n(n-1)}$, 
\begin{align*}
    T_k(n)
    &= \frac{
        \sum_{s=1}^{n-1} \left( \frac{n(n-1)}{\alpha s(n-s)} \cdot m(s,k) \cdot M(s,k) \right)
    }{n}
    \\
    &= \frac{n-1}{\alpha} \left(
        \sum_{s=1}^{k-1} \frac{s(n-k)}{s(n-s)}
        + \sum_{s=k}^{n-1} \frac{k(n-s)}{s(n-s)}
    \right)
    \\
    &= \frac{n-1}{\alpha} \left(
        (n-k) \sum_{s=1}^{k-1} \frac{1}{n-s}
        + k \sum_{s=k}^{n-1} \frac{1}{s}
    \right)
    \\
    &= \frac{n-1}{\alpha} \left(
        (n-k) \sum_{s=n-k+1}^{n-1} \frac{1}{s}
        + k \sum_{s=k}^{n-1} \frac{1}{s}
    \right)
    \\
    &= \frac{n-1}{\alpha} \left(
        (n-k) (H_{n-1} - H_{n-k})
        + k (H_{n-1} - H_{k-1})
    \right),
\end{align*}
where $H_j = \sum_{i=1}^{j} i^{-1}$ is the $j$-th harmonic number.

\end{proof}

\section{Equivalence of $\model_1$ and $\model_2$}\label{sec:coupling_unbiased_sync}
\begin{lemma}
\label{prop:equiv-models}
Assume $\model_1$ and $\model_2$ are initialized with the same 
distribution of opinions at time $0$. Let $\bx\in\{0, 1\}^n$. Then, for every $t\ge 0$:
\[
	\Prob{\model_1}{\bx^{(t)} = \bx} = \Prob{\model_2}{\bx^{(t)} = \bx}.
\]
\end{lemma}
\begin{proof}
Denote by $\bx^{(t)}$ the state of the system at time $t$ and assume $\bx^{(t)}
= \bx$, where $\bx\in\{0, 1\}^n$. In particular, suppose that $x_u = a\in\{0,
1\}$. Denote by $d_a$ and $d_{1-a}$ respectively the number of $u$'s neighbours
holding opinions $a$ and $1 - a$ at the end of iteration $t$. We begin by
showing that the probabilities of $u$ holding opinion $a$ in iteration $t+1$
are the same in the two models. We first consider the case $x_u = a$ (i.e., we
are interested in the probability that $u$ does not change opinion between
iterations $t$ and $t+1$). For $\model_1$ we have: 
\begin{align*}
	&\Prob{\model_1}{x_u^{(t+1)} = a\mid \bx^{(t)} = \bx} = 
	\frac{d_a}{d_u} + \frac{d_{1-a}}{d_u}(1 - \alpha) = 1 - 
	\alpha\frac{d_{1-a}}{d_u},
\end{align*}
where we used $d_a + d_{1-a} = d_u$. For $\model_2$ we have:
\begin{align*}
	&\Prob{\model_2}{x_u^{(t+1)} = a\mid \bx^{(t)} = \bx} = 
	(1 - \alpha) + \alpha\frac{d_a}{d_u} 
    = 1 - \alpha\left( 1 - \frac{d_{a}}{d_u} \right)\\
    &= 1 - \alpha\frac{d_{1-a}}{d_u},
\end{align*}
where the last equality follows since $d_u - d_a = d_{1-a}$.

We next consider the case $x_u = 1 - a$. We have:
\begin{align*}
	\Prob{\model_1}{x_u^{(t+1)} = a\mid \bx^{(t)} = \bx} &= \frac{d_{a}}{d_u}\alpha,
    \\
    \Prob{\model_2}{x_u^{(t+1)} = a\mid \bx^{(t)} = \bx} &= \alpha\frac{d_{a}}{d_u},
\end{align*}
We have thus shown that
\[
	\Prob{\model_1}{x_u^{(t)} = a \mid \bx^{(t-1)}=\bx} 
	= \Prob{\model_2}{x_u^{(t)} = a \mid \bx^{(t-1)}=\bx}.
\]
Moreover, since nodes' decisions are assumed \emph{independently} by 
each node  in every iteration both in $\model_1$ and $\model_2$, the above implies:
\begin{align*}
	&\Prob{\model_1}{\bx^{(t)} = \bx' \mid \bx^{(t-1)}=\bx} = \prod_{u\in 
	V}\Prob{\model_1}{x_u^{(t)} = x_u' \mid \bx^{(t-1)}=\bx},
\end{align*}
and the same of course holds for $\model_2$. As a consequence:
\begin{equation}\label{eq:equiv}
	\Prob{\model_1}{\bx^{(t)} = \bx'|\bx^{(t-1)}=\bx} = 
	\Prob{\model_2}{\bx^{(t)} = \bx' |\bx^{(t-1)}=\bx}
\end{equation}
We complete the proof by inductively showing that 
\begin{align*}
	&\Prob{\model_1}{\bx^{(t-1)} = \bx} = \Prob{\model_2}{\bx^{(t-1)} = 
	\bx}\\ 
	&\Longrightarrow \Prob{\model_1}{\bx^{(t)} = \bx} = \Prob{\model_2}{\bx^{(t)} = 
\bx}.
\end{align*}
To this purpose, note that our arguments above immediately imply that the claim
is true for $t = 1$, if $\model_1$ and $\model_2$ are initialized in the same
configuration $\bx^{(0)}$. For the inductive step we have:
\begin{align*}
	&\Prob{\model_1}{\bx^{(t)} = \bx}\\
	& = \sum_{\bz\in \{0, 1\}^V}\Prob{\model_1}{\bx^{(t)} = \bx \mid \bx^{(t-1)} = 
	\bz}\Prob{\model_1}{\bx^{(t-1)} = \bz}\\
	&= \sum_{\bz\in \{0, 1\}^V}\Prob{\model_2}{\bx^{(t)} = \bx \mid \bx^{(t-1)} = 
	\bz}\Prob{\model_2}{\bx^{(t-1)} = \bz},
\end{align*}
where the last equality follows from \eqref{eq:equiv} and from the 
inductive hypothesis. This concludes the proof.
\end{proof}

\section{Proof of Theorem~\ref{cl:consensus_1}}\label{apx:unbiased_sync_thm_proof}
\begin{proof}[Proof of part (i)]
Note that given the state vector $\mathbf x^{(t)}$, the probability that
$x^{(t+1)}_i=1$ can be expressed as follows: 

\[
\Prob{}{x^{(t+1)}_i = 1 \,|\, \mathbf x^{(t)}} =
(1 - \alpha) x_i^{(t)} + \alpha \frac{ \sum_{j \in V} a_{ij} x_j^{(t)} }{ d_i}, 
\]
since with probability $1-\alpha$ no opinion is copied to node $i$, and with
probability $\alpha$ an opinion is copied to node $i$ and the probability that
it is opinion 1 is exactly $\sum_j a_{ij} x_j^{(t)} / d_i$. 

Therefore, after defining $p_i(t) := \Prob{}{x_i^{(t)}=1 \,|\, \mathbf
x^{(t-1)}}$, we can write the model in the following form (called the
\emph{binary influence model} in \cite{asavathiratham2000influence}): 
\begin{align*}
\mathbf p(t+1) &= P \mathbf x^{(t)} \\
\mathbf x(t+1) &= \mathcal{B}(\mathbf p(t+1))
\end{align*}
where $P$ is a row-stochastic matrix, namely 
\[
P_{ij} = (1 - \alpha) \delta_{ij} + \alpha \frac{a_{ij}}{d_i}
\]
where $\delta_{ij}=1$ if $i=j$ and $\delta_{ij}=0$ otherwise, 
and $\mathcal{B}(\mathbf p(t))$ stands for a vector of $n$ Bernoulli random variables, 
respectively with parameters $p_1(t),\ldots,p_n(t)$. 
Observe that since each $x_i^{(t)}$ is binary, 
$
\mathbf p(t) = \Expec{}{\mathbf x^{(t)} \,|\, \mathbf x^{(t-1)}}. 
$
Thus, $\Expec{}{\mathbf x^{(t+1)} \,|\, \mathbf x^{(t)}} = \mathbf 
p^{(t+1)} = P \mathbf x^{(t)}$. 
Proceeding inductively we obtain 
\begin{align*}
&\Expec{}{\mathbf x^{(t)} \,|\, \mathbf x^{(0)}} = 
\Expec{}{ \Expec{}{ \mathbf x^{(t)} \,|\, \mathbf x^{(t-1)} } \,|\, \mathbf x^{(0)} } 
\\
&= \Expec{}{ P \mathbf x^{(t-1)} \,|\, \mathbf x^{(0)} } 
= P \Expec{}{ \mathbf x^{(t-1)} \,|\, \mathbf x^{(0)} }
= P^t \mathbf x^{(0)}.
\end{align*}
Since the operator $P$ is ergodic (being identical to the transition matrix of a lazy random walk on a connected graph), 
\[
\lim_{t \to \infty} \Expec{}{ \mathbf x^{(t)} \,|\, \mathbf x^{(0)} } = 
\lim_{t \to \infty} P^t \mathbf x^{(0)} = \mathbf 1 \mathbf \pi \mathbf x^{(0)},
\]
where $\mathbf \pi$ is the left dominant eigenvector of $P$, normalized so that
$\mathbf \pi \mathbf 1 = 1$. Note that $\pi$ represents the unique stationary
distribution associated to $P$. In particular, since $P$ is the transition
matrix of a lazy random walk on the graph, $\pi_i = {d_i}/{\sum_{j \in V} d_j}$
and since by assumption $x^{(0)}_i = 1$ if $i \in W$ and 0 otherwise, we get 
\[
\lim_{t \to \infty} \Expec{}{ \mathbf x^{(t)} \,|\, \mathbf x^{(0)} } 
= \mathbf 1 \mathbf \pi \mathbf x^{(0)} 
= \mathbf 1 \frac{\sum_{u \in W} d_u}{\sum_{u \in V} d_u}. 
\]
Finally, recalling that $\Expec{}{ x_u^{(t)} \,|\, \mathbf x^{(0)} }
	= \Prob{}{x_u^{(t)}=1 \,|\, \mathbf x^{(0)}}$ since the $x_u^{(t)}$ 
	are binary random variables, we get $\phi = (\sum_{u \in W} 
	d_u)/(\sum_{u \in V} d_u)$.
\end{proof}

\begin{proof}[Proof of part (ii)]
We can achieve a bound on the expected consensus time by leveraging results
that appear in a number of papers and are condensed in~\cite[Section
14.3.2]{aldous2002reversible} for the case of continuous-time random walks,
while literature on the synchronous, discrete setting is sparser. For this
reason, in the remainder we retrace the main points of the proof, adapting it
to the discrete-time, synchronous case.  To this purpose, we need some
additional notation.  Following~\cite{aldous2002reversible}
and~\cite{Levin:2009}, considered a node/state $v$ of a (henceforth, ergodic
and reversible) Markov chain, we denote by $\thit_v$ the \emph{hitting time} of
$v$, i.e., the number of steps till the chain reaches $v$ for the first time.
We denote by $\Expec{u}{\thit_v}$ the expected hitting time of $v$ when the
random walk starts at node $u$. Recall that the hitting time for a given graph
is the maximum of $\Expec{u}{\thit_v}$ over all possible choices of $u$ and
$v$, i.e., $\thit := \max_{u,v}\Expec{u}{\thit_v}$. Considered some probability
distribution $\mathcal{D}$ over the states of the Markov chain, with a slight
abuse of notation, we denote by $\Expec{\mathcal{D}}{\thit_v}$ the expected
hitting time of $v$ when the random walk starts at $u\sim\mathcal{D}$.
Moreover, we denote by $M_{uv}$ the \emph{meeting time} of two independent
copies of the Markov chain started at $u$ and $v$ respectively. Finally, the
correspondence between our voter-like model and lazy random walks allows us to
conclude that the consensus time $\tcons$ follows the same distribution as the
\emph{coalescing time} of the corresponding, lazy random walk (for the reasons
behind this fact, refer to~\cite[Section 14.3]{aldous2002reversible}
or~\cite{hassin01voter}).  This time is defined with respect to a coalescing
(possibly lazy) random walk over $G$: at time $0$, we start $n$ independent
random walks, one per node/state in $G$. Assume there are $x$ surviving walks
at the end of step $t$; then, in step $t+1$, each of the $x$ walks moves to a
random neighbor, independently of the others; if two (or more) walks move to
the same node in step $t+1$, they stick together thereafter, moving as a single
one. Time $\tcons$ is defined as the first step in which there is only one
surviving walk. In light of the above considerations, in the remainder of this
proof, we use both the terms ``consensus time'' and ``coalescing time'' to
refer to $\tcons$.

\begin{lemma}\label{le:TvsM}
For a discrete, ergodic and reversible Markov chain, for every $u, v\in V$ we have:
\begin{equation*}\label{eq:TvsM}
	\max_{u,v}\Expec{}{M_{uv}}\le\max_{u,v}\Expec{u}{\thit_v}.
\end{equation*}
\end{lemma}
\begin{proof}
We follow the very same lines as the proof
of~\cite[Proposition~14.5]{aldous2002reversible}, which is given for continuous
Markov chains. First of all, for $u,v \in V$, let $X_t$ and $Y_t$ be the
chains at time $t$ respectively started at $u$ and $v$.  For $x,y \in V$, we
define the following function:
\[
	f(x, y) = \Expec{x}{T_y} - \Expec{\pi}{T_y}.
\]
Next, we define the following random variable:
\begin{align}\label{eq:martingale}
    S_t = \begin{cases}
        2t + f(X_t, Y_t) & \text{if } 0\le t\le M_{uv},
        \\
        S_{t-1} & \text{if } t > M_{uv}.
    \end{cases}
\end{align}
\paragraph{The $S_t$'s form a martingale.}
It should be noted that here, both the length of the sequence (i.e., $M_{uv}$)
and the values of the $S_t$'s depend on the randomness of two walks started at
$u$ and $v$ respectively.  To prove this first result, we proceed in a way
similar to the proof of \cite[Proposition 3.3]{aldous2002reversible}.  Clearly,
we have $\Expec{}{S_t \mid S_{1},\ldots , S_{t-1}} = S_{t-1}$, whenever $t >
M_{uv}$. Next, assume $(X_i, Y_i) = (x_i, y_i)$, for $i = 1,\ldots, t-1$, with
$x_i\ne y_i$ for every $i$.  We first note that it is enough to prove that 
\begin{align*}
	\Expec{}{S_t \mid \bigcap_{i=1}^{t-1} (X_i, Y_i) = (x_i, y_i)} &= 
    2t + \Expec{}{f(X_t, Y_t) \mid \bigcap_{i=1}^{t-1} (X_i, Y_i) = (x_i, y_i)}\\
	& = 2t + \Expec{}{f(X_t, Y_t) \mid (X_{t-1}, Y_{t-1}) = (x_{t-1}, y_{t-1})}
\end{align*}
where the second equality follows from the Markov property, which 
clearly also applies to $f(X_t, Y_t)$. We therefore have, if $P$ is the 
transition matrix of a reversible Markov chain:
\begin{align*}
	\Expec{}{S_t | S_{1},\ldots , S_{t-1}} 
    &= 2t + \sum_{x,y}P(x_{t-1}, x) P(y_{t-1}, y)f(x, y) 
    = 2t + \sum_xP(x_{t-1}, x)\sum_yP(y_{t-1}, y)f(x, y)
    \\
	&\stackrel{(a)}{=} 2t + \sum_xP(x_{t-1}, x)(f(x, y_{t-1}) - 1) 
    = 2t + \sum_xP(x_{t-1}, x)f(x, y_{t-1}) - 1 
    \\
    &\stackrel{(a)}{=} 2t + f(x_{t-1}, y_{t-1}) - 2 
	= S_{t-1},
\end{align*}
where $(a)$ follows since, for $x\ne y$, by the one-step recurrence of $f(x, y)$ 
we get $f(x, y) = 1 + \sum_zP(x, z)f(z, y)$ 
(see \cite[Proposition 3.3]{aldous2002reversible}).

\paragraph{Expectation of meeting time.}
Consider any two states $u$ and $v$. We have from the definition of 
$S_t$:
\[
	S_{M_{uv}} = 2M_{uv} + f(X_{M_{uv}}, Y_{M_{uv}}), 
\]
hence:
\begin{align*}
	&\Expec{}{S_{M_{uv}}} = 2\Expec{}{M_{uv}} + \Expec{(X_t, 
	Y_t)}{\Expec{X_{M_{uv}}}{T_{Y_{M_{uv}}}} - 
	\Expec{\pi}{T_{Y_{M_{uv}}}}} = 2\Expec{}{M_{uv}} - \Expec{(X_t, 
	Y_t)}{\Expec{\pi}{T_{Y_{M_{uv}}}}},
\end{align*}
where the second equality follows since, for every realization of 
$(X_t, Y_t)$, we have $X_{M_{uv}}\equiv Y_{M_{uv}}$ by definition of 
$M_{uv}$. Next,
\begin{align*}
	&\Expec{(X_t, Y_t)}{\Expec{\pi}{T_{Y_{M_{uv}}}}} = 
	\sum_y\Expec{\pi}{T_y}\Prob{}{\text{$X_t$ and $Y_t$ meet at 
	$y$}}\le\max_{ij}\Expec{i}{T_j}, 
\end{align*}
which implies:
\begin{equation}\label{eq:meet_low}
	\Expec{}{S_{M_{uv}}}\ge 2\Expec{}{M_{uv}} - 
	\max_{ij}\Expec{i}{T_j}.
\end{equation}
On the other hand, since $S_t$ is a martingale, we can apply the 
optional stopping theorem, whence:
\begin{equation}\label{eq:meet_up}
	\Expec{}{S_{M_{uv}}} = \Expec{}{S_0} = \Expec{}{f(u, v)} = 
	\Expec{u}{T_v} - \Expec{\pi}{T_v}
\end{equation}
Putting together \eqref{eq:meet_low} and \eqref{eq:meet_up} and
noting that $\Expec{u}{T_v}\le\max_{ij}\Expec{i}{T_j}$, 
we conclude that:
\[
	2\Expec{}{M_{uv}}\le\Expec{}{S_{M_{uv}}} + 
	\max_{ij}\Expec{i}{T_j} = \Expec{u}{T_v} - \Expec{\pi}{T_v} + 
	\max_{ij}\Expec{i}{T_j}\le 2\max_{ij}\Expec{i}{T_j},
\]
whence the thesis of Lemma~\ref{le:TvsM}.
\end{proof}

Using the aforementioned, well-known equivalence between (lazy) coalescing
random walks and the voter model (see~\cite{hassin01voter} or~\cite[Section
14.3]{aldous2002reversible}), we next provide an upper bound on
$\Expec{}{\tcons}$, which also provides a bound on the convergence time of the
voter model. Note that this result is only provided for the sake of
completeness and holds for any reversible Markov chain. 

First of all, we fix an arbitrary ordering of the nodes of $G$, let it be
$u_1,\ldots, u_n$, and we label the random walks accordingly. We then consider
the following, equivalent formulation of the coalescing random walk process:
whenever two or more random walks meet at step $t$, they coalesce and follow
thereafter the future path of the lower-labeled random walk.  Note that we
deterministically have $\tcons \ge\max_{v}M_{u_1v}$, possibly with 
strict inequality.

Next, let $m = \max_{u,v}\Expec{}{M_{u, v}}$ and note that $m\le \thit$ from 
Lemma \ref{le:TvsM}. 
We have:
\begin{align*}
	\Expec{}{\tcons}&=\sum_{t=0}^\infty\Prob{}{\tcons \ge t}
	\le\sum_{t=0}^\infty\Prob{}{\max_{v}M_{u_1v}\ge t}
	\\
	&\le\sum_{t=0}^\infty\min\left\{1, \sum_{v\ne u_1}\Prob{}{M_{u_1v}\ge t}\right\}.
\end{align*}
On the other hand, by proceeding as in \cite[Section 
2.4.3]{aldous2002reversible}) to derive (2.20) we obtain, for 
every $v \in V$:
$
	\Prob{}{M_{u_1v}\ge	t}\le e^{-\lfloor\frac{t}{m}\rfloor},
$
whence:
\begin{align*}
	\Expec{}{\tcons}
	&\le\sum_{t=0}^\infty\min\left\{1, ne^{-\lfloor\frac{t}{m}\rfloor}\right\} 
	\le \sum_{t=0}^\infty\min\left\{1, ne^{-\frac{t-1}{m}}\right\}
	\\
	&< 1+ \int_{0}^\infty\min\left\{1, ne^{-\frac{t-1}{m}}\right\} dt
	\\
	&\stackrel{(a)}{=} 1+ \int_{0}^{1+m\ln n} 1 dt + n\int_{1+m\ln n}^{\infty} e^{-\frac{t-1}{m}} dt
	\\
	&= 2 + m\ln n + m 
	\le (\ln n + 3)\thit,
\end{align*}
where $(a)$ follows from the fact that $ne^{-\frac{t-1}{m}}<1$ for $t>1+m \ln n$.

\paragraph{Remark for the case $\alpha\le 1/2$.}
Note that whenever $\alpha\le 1/2$, it is possible to get a stronger bound
on the expected consensus time. In the remainder, we write the transition matrix
of the process in compact form as:
$
	P = (1 - \alpha)I + \alpha Q,
$
where $Q$ is the transition matrix associated to a random walk on $G$, 
namely, $Q_{ij} = \frac{a_{ij}}{d_i}$. We denote by $\lambda_1 
\ge\lambda_2\ge\cdots\ge\lambda_n$ the eigenvalues of $Q$.
Being $G$ undirected, $Q$ and $P$ define reversible Markov chains. 
Moreover, if $\bw$ is a right eigenvector of $Q$ with eigenvalue 
$\lambda$:
\[
	P\bw = (1 - \alpha)\bw + \alpha Q\bw = (1 - \alpha + 
	\alpha\lambda)\bw,
\]
i.e., $\bw$ is also a right eigenvector of $P$ with eigenvalue 
$\lambda' = 1 - \alpha + \alpha\lambda$. In particular, we have 
$\lambda'\ge 0$, whenever
$
	\alpha\le {1}/({1 - \lambda}).
$
Since $Q$ is stochastic, we have $\lambda_i \in [-1,1]$ for every $i$.
Therefore, whenever $\alpha\le\frac{1}{2}$, we also have $\lambda_i'\ge 0$, 
for every $i$, namely $P$ is positive semidefinite (PSD). In this case, we can
apply a number of recent results for lazy random walks \cite{cooper_coalescing_2013,kanade_coalescence_2019,oliveira2019random}.
In particular, \cite[Theorem 6]{oliveira2019random}
immediately implies that the expected time of convergence to consensus, 
which corresponds to the coalescing time of $n$ coalescing random walks 
on $G$ satisfies:
\[
	\Expec{}{\tcons}\le c \, \thit,
\]
where $c$ is a universal constant and $\thit$ is the largest expected 
hitting time on $G$.
\end{proof}

\section{Biased fixation probability on regular graphs}
\label{sec:async-fixation-reg}
The analysis from the subsection \ref{se:biased} can be extended from the
clique to any $\Delta$-regular graph as follows \cite{diaz2016absorption}. Let
$S$ be the set of nodes with opinion 1 at any given step. The transition
probabilities of the birth-and-death chain on $\{0,1,\ldots,n\}$ will now
generally depend on $S$, and not just on the size $k=|S|$. For example, if
$p_S$ is the probability of transitioning from the set $S$ of nodes with
opinion 1 to a set of size of $|S|+1$, 

\begin{align*}
p_S &= \sum_{u \notin S} \Prob{}{u} \Prob{}{v \in S \,|\, u} \alpha_{01} \\
&= \sum_{u \notin S} \frac{1}{n} \frac{|N(u) \cap S|}{\Delta} \alpha_{01} \\
&= \frac{\alpha_{01}}{n \Delta} |\partial S|, 
\end{align*}
where $N(u)$ is the set of nodes adjacent to $u$ and $\partial S$ is the cut
through $S$. 

Similarly, if $q_S$ is the probability of transitioning from the set $S$  of
nodes with opinion 1 to a set of size $|S|-1$, 
\begin{align*}
q_S &= \frac{\alpha_{10}}{n \Delta} |\partial \bar{S}| = \frac{\alpha_{10}}{n \Delta} |\partial S|. 
\end{align*}
The key observation is that the ratio $p_S/q_S = \alpha_{01}/\alpha_{10} = r$
equals the fitness and does not depend on the set $S$ at all. Therefore, the
process can still be cast as a birth-and-death process if we ignore all steps
in which the number of nodes with opinion $1$ does not change: the presence of
such steps is irrelevant for the fixation probability. Therefore, by the same
analysis of the previous section, we obtain a fixation probability of $(1-
r^{-k})/(1- r^{-n})$ when starting from a configuration with $k$ nodes holding
opinion $1$.

\section{Proof of Theorem \ref{thm:async-biased-clique-ect}}
\label{sec:biased-async-appendix}
Define the matrix $B$ as in the proof of Theorem~\ref{thm:exact-clique}; note
that in this setting $p_i = r q_i$ for all $i=1,\ldots,n-1$, where $r =
\alpha_{01}/\alpha_{10}$.  Recall the formulation for the expected consensus
time starting from $k$ nodes with opinion 1, namely the $k$-th entry of the
vector $T(n)=B^{-1}\mathbf{1}$, which by proceeding as in the proof of
Lemma~\ref{lem:B-inverse} can be computed as 
\begin{align*}
    T_k(n) 
    &= \frac{
        \sum_{s=1}^{n-1} q_s^{-1} \left(
            \sum_{\ell=1}^{m(s,k)} r^{s-\ell}
        \right) \left(
            \sum_{\ell=1}^{M(s,k)} r^{\ell-1}
        \right)
    }{
        \sum_{\ell=1}^{n} r^{\ell-1}
    }, 
\end{align*}
with $m(s,k) := \min(s,k)$, and $M(s,k) := n - \max(s,k)$. 

The above expression for $T_k(n)$, although exact, is very hard to understand
asymptotically, therefore we proceed to bound it asymptotically with simpler
expressions. For the remainder of the proof we assume $r < 1$
(i.e.~$\alpha_{01}<\alpha_{10}$); this is without loss of generality, up to a
relabeling of the opinions. 

Using a difference of geometric sums, we can write
\[
    \sum_{\ell=1}^{\min(s,k)} r^{s-\ell}
    = \sum_{\ell=0}^{s-1} r^{\ell} - \sum_{\ell=0}^{\max(-1,s-k-1)} r^{\ell}
    = \frac{1-r^s}{1-r} - \frac{1-r^{\max(0,s-k)}}{1-r}
    = \frac{r^{\max(0,s-k)}-r^s}{1-r}.
\]
Therefore, substituting $q_s = \alpha_{10} \frac{s(n-s)}{n(n-1)}$ and computing the other geometric sums, 
\begin{align*}
    T_k(n) 
    &= \frac{
        \sum_{s=1}^{n-1} \frac{n(n-1)}{\alpha_{10}s(n-s)} 
        \left( \frac{r^{\max(0,s-k)}-r^{s}}{1-r} \right) \left( \frac{1-r^{M(s,k)}}{1-r} \right)
    }{
        \frac{1-r^{n}}{1-r}
    }
    \\
    &= \frac{n(n-1)}{\alpha_{10}(1-r)(1-r^{n})} 
        \sum_{s=1}^{n-1} \frac{(r^{\max(0,s-k)}-r^{s})(1-r^{M(s,k)})}{s(n-s)}.
\end{align*}
Note that the $\max(\cdot)$ in the exponents (also hidden in $M(s,k) = n-\max(s,k)$) 
can be removed by splitting the sum into two parts at $s=k$. In particular, we also have
\[
    T_k(n) = \frac{n(n-1)}{\alpha_{10}(1-r)(1-r^{n})} 
    \left(
        \sum_{s=1}^{k-1} \frac{(1-r^{s})(1-r^{n-k})}{s(n-s)}
        + \sum_{s=k}^{n-1} \frac{(r^{s-k}-r^{s})(1-r^{n-s})}{s(n-s)}
    \right).
\]
Note that $T_k(n) = \frac{n(n-1)}{\alpha_{10}(1-r)(1-r^n)} (S_1-S_2)$, where
\begin{align*}
    S_1 &= (1-r^{n-k})\left[
        \sum_{s=1}^{k-1} \frac{1}{s(n-s)}
        - \sum_{s=1}^{k-1} \frac{r^{s}}{s(n-s)}
    \right]
    \\
    S_2 &= (r^{-k}-1) \left[
        r^n\sum_{s=k}^{n-1} \frac{1}{s(n-s)}
        - \sum_{s=k}^{n-1} \frac{r^{s}}{s(n-s)}
    \right]
\end{align*}
For the first term $S_1$ above:
\begin{align*}
    &(1-r^{n-k})\left[
        \sum_{s=1}^{k-1} \frac{1}{s(n-s)}
        - \sum_{s=1}^{k-1} \frac{r^{s}}{s(n-s)}
    \right]
    = \frac{1-r^{n-k}}{n}\left[
        H_{n-1}+(H_{k-1}-H_{n-k})
        - n\sum_{s=1}^{k-1} \frac{r^{s}}{s(n-s)}
    \right],
\end{align*}
where we used 
\begin{flalign*}
    \sum_{s=1}^{k-1}\frac{1}{s(n - s)} 
    &= \frac{1}{n}\sum_{s=1}^{k-1}\left(\frac{1}{s} + \frac{1}{n-s}\right)
    = \frac{1}{n}\left(\sum_{s=1}^{k-1}\frac{1}{s} + \sum_{s=n-k+1}^{n-1}\frac{1}{s}\right) &&
    \\
    &= \frac{1}{n}[H_{n-1}+(H_{k-1}-H_{n-k})]. 
\end{flalign*}
For the second term $S_2$ above:
\begin{align*}
    &(r^{-k}-1) \left[
        r^n\sum_{s=k}^{n-1} \frac{1}{s(n-s)}
        - \sum_{s=k}^{n-1} \frac{r^{s}}{s(n-s)}
    \right]
    = \frac{r^{-k}-1}{n} \left[
        r^n[H_{n-1}-(H_{k-1}-H_{n-k})]
        - n\sum_{s=k}^{n-1} \frac{r^{s}}{s(n-s)}
    \right],
\end{align*}
where we used 
\begin{flalign*}
    \sum_{s=k}^{n-1}\frac{1}{s(n - s)}
    &= \frac{1}{n}\sum_{s=k}^{n-1}\left(\frac{1}{s} + \frac{1}{n-s}\right)
    = \frac{1}{n}\left[
        \sum_{s=1}^{n-1}\left(\frac{1}{s} + \frac{1}{n-s}\right)
        - \sum_{s=1}^{k-1}\left(\frac{1}{s} + \frac{1}{n-s}\right)
    \right] &&
    \\
    &= \frac{1}{n}[2H_{n-1}-H_{n-1}-(H_{k-1}-H_{n-k})] 
    \\
    &= \frac{1}{n}[H_{n-1}-(H_{k-1}-H_{n-k})]
\end{flalign*}
Putting the two equalities for $S_1$ and $S_2$ together,
\begin{small}
\begin{align*}
    n(S_1 - S_2) 
    &= (1-r^{n-k})\left[
        H_{n-1}+(H_{k-1}-H_{n-k} )
        - n\sum_{s=1}^{k-1} \frac{r^{s}}{s(n-s)}
    \right] 
    +\\&\quad
    - (r^{-k}-1)\left[
        r^n(H_{n-1}-(H_{k-1}-H_{n-k}))
        - n\sum_{s=k}^{n-1} \frac{r^{s}}{s(n-s)}
    \right] 
    \\
    &= (1-2r^{n-k}+r^n)H_{n-1} + (1-r^n)(H_{k-1}-H_{n-k})
    +\\&\quad
    - (1-r^{n-k}) n\sum_{s=1}^{k-1} \frac{r^{s}}{s(n-s)}
    + (r^{-k}-1) n\sum_{s=k}^{n-1} \frac{r^{s}}{s(n-s)}
    \\
    &= (1-2r^{n-k}+r^n)H_{n-1} + (1-r^n)(H_{k-1}-H_{n-k}) 
    +\\&\quad
    - (1-r^{n-k}) \sum_{s=1}^{k-1} \left[ \frac{r^{s}}{s} + \frac{r^{s}}{n-s} \right]
    + (r^{-k}-1) \sum_{s=k}^{n-1} \left[ \frac{r^{s}}{s} + \frac{r^{s}}{n-s} \right]
    \\
    &= (1-2r^{n-k}+r^n)H_{n-1} + (1-r^n)(H_{k-1}-H_{n-k}) 
    +\\&\quad
    - 2\sum_{s=1}^{n-1} \frac{r^{s}}{s} 
    + r^{n-k} \sum_{s=1}^{k-1} \left[ \frac{r^{s}}{s} + \frac{r^{s}}{n-s} \right]
    + r^{-k} \sum_{s=k}^{n-1} \left[ \frac{r^{s}}{s} + \frac{r^{s}}{n-s} \right]
    \\
    &= (1-2r^{n-k}+r^n)H_{n-1} + (1-r^n)(H_{k-1}-H_{n-k}) 
    +\\&\quad
    - 2(1-r^{n-k})\sum_{s=1}^{n-1} \frac{r^{s}}{s} 
    + r^{-k}(1-r^n) \sum_{s=k}^{n-1} \left[ \frac{r^{s}}{s} + \frac{r^{s}}{n-s} \right].
\end{align*}
\end{small}

Note that, using the assumption $r<1$, 
\begin{itemize}
    \item $(1-2r^{n-k}+r^n)H_{n-1} + (1-r^n)(H_{k-1}-H_{n-k}) 
        < 2H_{n-1} < 2\log n + 2$
    \item $-2(1-r^{n-k})\sum_{s=1}^{n-1} \frac{r^{s}}{s} 
        < -2(1-r)r $
    \item $r^{-k}(1-r^n) \sum_{s=k}^{n-1}\left[ \frac{r^{s}}{s} + \frac{r^{s}}{n-s} \right]
        < 2r^{-k}(1-r^n) \sum_{s=k}^{n-1} r^{s}
        = 2r^{-k}(1-r^n) \frac{r^k-r^n}{1-r}
        < \frac{2(1-r^n)}{1-r}$
\end{itemize}

Hence it follows that $n(S_1-S_2) = O(\log n)$.
Therefore, when $r$ and $\alpha_{10}$ are constant with respect to $n$
and for any value of $k$ it holds:
\[
    T_k(n) = \frac{n(n-1)}{\alpha_{10}(1-r)(1-r^n)} (S_1-S_2) = O(n \log n).
\]

Note that such a bound is asymptotically tight for some values of $k$.
In fact, let us consider $k=\frac{n+1}{2}$ (for odd values of $n$). 
Note that:
\begin{itemize}
    \item $(1-2r^{n-k}+r^n)H_{n-1} + (1-r^n)(H_{k-1}-H_{n-k}) 
        = (1-2r^{n-k}+r^n)H_{n-1} > (1-2r^{\frac{n-1}{2}})H_{n-1}$
    \item $-2(1-r^{n-k})\sum_{s=1}^{n-1} \frac{r^{s}}{s} 
        > -2(1-r^{n-k})\sum_{s=1}^\infty\frac{r^s}{s} = 2(1-r^n)\ln(1-r)$
    \item $r^{-k}(1-r^n) \sum_{s=k}^{n-1}\left[ \frac{r^{s}}{s} + \frac{r^{s}}{n-s} \right]
        > 2r^{-k}(1-r^n) \frac{r^k}{k}
        > \frac{2(1-r^n)}{n}$
\end{itemize}
Hence it follows that there exists $k$ (namely $k=\frac{n+1}{2}$) such that 
$n(S_1-S_2) = \Omega(\log n)$.
Therefore, when $r$ and $\alpha_{10}$ are constant with respect to $n$,
for $k=\frac{n+1}{2}$ it holds:
\[
    T_k(n) = \frac{n(n-1)}{\alpha_{10}(1-r)(1-r^n)} (S_1-S_2) = \Omega(n \log n).
\]

\end{document}